\documentclass[letterpaper, 10 pt, conference]{ieeeconf}  % Comment this line out if you need a4paper

\IEEEoverridecommandlockouts                              % This command is only needed if 
\usepackage{amsmath} % assumes amsmath package installed
\usepackage{amssymb}  % assumes amsmath package installed

\usepackage{mleftright}

\newtheorem{theorem}{Theorem}
\newtheorem{proposition}{Proposition}
\newtheorem{lemma}{Lemma}
\newtheorem{corollary}{Corollary}

\newtheorem{remark}{Remark}

\newcommand{\R}{\mathbb R}

\newcommand{\dfb}{\stackrel{\Delta}{=}}

\newcommand{\gensn}[1]{|#1|} % Notation for general consensus seminorm (on matrices)
\newcommand{\vecsn}[1]{|#1|} % Notation for general consensus norm (on column vectors)
\newcommand{\restnorm}[1]{\|#1\|_R} % Notation for matrix norm on restricted (quotient) space
\DeclareMathOperator{\coe}{coe} % Notation for Markov-Dobrushin coefficient of ergodicity

\newcommand{\consubsp}{\mathcal{I}} % Notation for space spanned by a column vector of ones
\newcommand{\eqrowsum}{\mathcal{M}} % Notation for space containing all matrices of the form 1c'
\newcommand{\stoc}{\mathbb{S}} % Notation for set of all stochastic matrices
\newcommand{\stocpd}{\mathcal{R}} % Notation for set of all stochastic matrices with all diagonal entries positive and rooteg graphs
\newcommand{\stocsary}{\mathcal{K}} % Notation for set of all Sarysmsakov matrices
\newcommand{\stocscr}{\mathcal{S}} % Notation for set of all scrambling matrices
\newcommand{\dblstoc}{\mathbb{S}_D} % Notation for set of all stochastic matrices
\newcommand{\dblstocpd}{\mathcal{D}} % Notation for set of all doubly stochastic matrices with positive diagonal entries and weakly connected graphs

\newcommand{\jsr}{\rho}
\newcommand{\semijsr}{\sigma}

\title{\LARGE \bf
On Determining the Convergence Rate of an Infinite Product of Stochastic Matrices*
}

\author{Ron Ofir and A. Stephen Morse% <-this % stops a space
\thanks{*This work was supported in part by the Air Force Office of Scientific Research, under award numbers FA9550-23-1-0175 and FA9550-25-1-0223. The work of R. Ofir was partially supported by the Viterbi Fellowship, Technion.}% <-this % stops a space
\thanks{Both authors are with the Department of Electrical and Computer Engineering, Yale University, CT, USA ({\tt\small \{ron.ofir,as.morse\}@yale.edu})}%
}

\begin{document}

\maketitle
\thispagestyle{empty}
\pagestyle{empty}

%%%%%%%%%%%%%%%%%%%%%%%%%%%%%%%%%%%%%%%%%%%%%%%%%%%%%%%%%%%%%%%%%%%%%%%%%%%%%%%%
\begin{abstract} 

By a convergent set is meant a set of stochastic matrices where every infinite product of matrices from every compact subset converges to a rank one matrix. Well-known examples include the set of all scrambling matrices, the set of all stochastic matrices with all diagonal entries positive and a rooted graph, the set of all Sarymsakov matrices, and the set of doubly stochastic matrices with positive diagonal entries and a weakly connected graph. It is known that every infinite product from each compact set of every convergent set converges to its limit exponentially fast, but not much is known about the rate of convergence when not all matrices involved are scrambling matrices.
This paper deals with bounding the rate of convergence in convergent sets using submultiplicative seminorms. It is shown that only in some convergent sets all matrices are contractions in the same seminorm, and in particular that this method cannot be used to determine the convergence rate for the class of matrices with positive diagonal entries and a rooted graph. As a second contribution, it is shown that for every compact convergent set and every submultiplicative seminorm, there is a finite number~$k$ such that all products of~$k$ matrices from the set are contractions in the seminorm. Finally, several open questions are posed for future research.

\end{abstract}

%%%%%%%%%%%%%%%%%%%%%%%%%%%%%%%%%%%%%%%%%%%%%%%%%%%%%%%%%%%%%%%%%%%%%%%%%%%%%%%%
\section{INTRODUCTION}

Consensus protocols are key to many distributed control, estimation, and optimization algorithms ~\cite{Bullo2009DistRobots,Jadbabaie2003Coordination,Nedic2009DistributedSubgradient,OlfatiSaber2007ConsensusAndCooperation,Tsitsiklis1984Thesis}. 
The analysis of such algorithms typically involves the study of the convergence properties of infinite products of stochastic matrices from some subset of the set of all $n \times n$ stochastic matrices $\stoc^{n \times n}$. Two questions invariably arise: 1. Does the infinite product of interest converge? 2. If the infinite product converges, at what rate does convergence take place? The first question is typically addressed by exploiting subsets of $\stoc^{n \times n}$ which are known to be ``convergent.'' A subset~$\mathcal{C} \subset \stoc^{n \times n}$ is ~\emph{convergent} if for each compact subset~$\mathcal{\bar C} \subset \mathcal{C}$, every product of infinitely many matrices from~$\mathcal{\bar C}$ converges to a rank one matrix, i.e., for every sequence~$S_0,S_1,\ldots \in \mathcal{\bar C}$,
\begin{equation*}
    \lim_{t \to \infty} S_t \cdots S_0 = \mathbf{1} c
\end{equation*}
for some~$c \in \R^{1 \times n}$. There are several well-known classes within~$\stoc^{n \times n}$ which are convergent including the set of all scrambling matrices\footnote{By a~\emph{scrambling matrix} is meant a stochastic matrix with no pair of orthogonal rows.}~$\stocscr$~\cite{Hajnal1958}, the set~$\stocpd$ of all stochastic matrices with positive diagonals and rooted graphs\footnote{A directed graph is \emph{rooted} if it contains a spanning tree.}~\cite{Cao2008ConsensusGraphical}, the set of all Sarymsakov matrices~$\stocsary$~\cite{Xia2019GenSary}, and the set~$\dblstocpd$ of all doubly stochastic matrices with positive diagonals and weakly connected graphs~\cite{Liu2011DetGossip}. It is known that~$\stocpd$ is actually the largest convergent subset of the set of all~$n \times n$ stochastic matrices with positive diagonals~\cite{Cao2008ConsensusGraphical}, and that~$\dblstocpd$ is the largest convergent subset of the set of all~$n \times n$ doubly stochastic matrices with positive diagonals~\cite{Liu2011DetGossip}. It is also known that~$\dblstocpd \subset \stocpd \subset \stocsary$.

One way to obtain a convergence rate for any infinite product of stochastic matrices from a compact subset~$\mathcal{\bar C}$ is to use a 
submultiplicative seminorm defined on a suitably defined subspace of $\mathbb{R}^{n\times n}$ containing~$\mathcal{C}$\footnote{A seminorm has the same defining requirements as a matrix norm, except that the seminorm of a matrix equaling zero need not imply that the matrix is zero.}. For such an approach to work, every matrix in~$\mathcal{C}$ would have to have a seminorm value less than 1 (i.e., be a contraction in the seminorm). Such seminorms exist for~$\stocscr$ and~$\dblstocpd$. For example, consider the well-known \emph{coefficient of ergodicity} of any matrix in~$M$ in the subspace~$\mathcal{M}$ of all matrices in~$\R^{n \times n}$ with equal row sums; that is
\begin{equation}\label{eq:coe}
    \coe(M) \dfb .5 \max_{i,j \in \mathbf{n}} \sum_{k \in \mathbf{n}} |M_{ik} - M_{jk}|,
\end{equation}
where~$\mathbf{n} = \{1,\dots,n\}$. It turns out that~$\coe(M)$ is also the value at~$M$ of a suitably defined submultiplicative seminorm~$\gensn{\cdot}_\infty$ defined on~$\mathcal{M}$; moreover~$\gensn{S}_\infty < 1$ if and only if~$S \in \stocscr$~\cite{Ofir2025ConsensusSeminorms,DePasquale2024seminorms}. A similar situation holds for~$\dblstocpd$. In this case it is known that the second largest singular value of any doubly stochastic matrix in~$\mathcal{D}$ is less than one~\cite{Liu2011DetGossip}; in addition for any matrix~$S \in \stoc^{n \times n}$ the second largest singular value of~$M$ is known to be the value at~$M$ of a suitably defined submultiplicative seminorm~$\gensn{\cdot}_2$ defined on~$\mathcal{M}$~\cite{Liu2011DetGossip}.

In light of the convergent sets and seminorms mentioned above, it is natural to expect that for every convergent set there is a submultiplicative consensus seminorm such that all matrices in the set are contractions. It was shown in~\cite{Ofir2025ConsensusSeminorms} that this is not the case. To be exact, it was shown that there is no convergent set larger than the set of scrambling matrices in the partial order induced by inclusions such that all matrices in the set are contractions in the same submultiplicative consensus seminorm. Here, we extend this result and prove the following.
\begin{theorem}\label{thm:no_R_K_contraction}
    Fix~$n > 2$. There is no submultiplicative consensus seminorm such that~$\gensn{S} < 1$ for all~$S \in \stocpd \subset \stocsary$.
\end{theorem}
In other words, there is no one seminorm which can be used to bound the convergence rate of \emph{every} infinite product of matrices from~$\stocpd$ or~$\stocsary$. 

A different way to obtain bounds on convergence rate of infinite products of matrices from a convergent set is to use the well-known fact that there is a finite number~$k$, which depends only the dimension~$n$, such that for every convergent set~$\mathcal{C}$ the product of each~$k$ matrices from~$\mathcal{C}$ is scrambling~\cite{Wolfowitz1963SIA}. Since the convergence rate for scrambling matrices can be conveniently bounded using the coefficient of ergodicity, this makes it possible to bound the convergence in general convergent sets by studying products of finite but non-trivial length. Still, this approach has been used for example in~\cite{Cao2008ConsensusRate} to derive a bound on the convergence rate for infinite products of so-called flocking matrices. Our second contribution is a preliminary result aimed at generalizing this approach to general consensus seminorms, proving that for every compact convergent set and every submultiplicative consensus seminorm, there is a finite number~$k$ such that the product of every~$k$ matrices from the set is a contraction.

\section{Consensus seminorms}

This section provides a brief review of consensus seminorms. Recall that a seminorm has the same properties as a norm except that it may take zero values on nonzero vector. Let~$\eqrowsum^{n \times n} \subset \R^{n \times n}$ denote the subspace consisting of all $n \times n$ real matrices whose row sums are all equal. We will call a seminorm~$\gensn{\cdot} : \eqrowsum^{n \times n} \to \R$ a~\emph{consensus seminorm} if~$\gensn{M} = 0 \iff M = \mathbf{1}c$ for some~$c \in \R^n$. The reason to restrict the domain of a consensus seminorm to the subspace~$\eqrowsum^{n \times n}$ rather than all of~$\R^{n \times n}$ is explained in detail in~\cite{Ofir2025ConsensusSeminorms}. A consensus seminorm is said to be submultiplicative if~$\gensn{M_1 M_2} \le \gensn{M_1} \gensn{M_2}$ for all~$M_1,M_2 \in \eqrowsum^{n \times n}$. The triangle inequality implies that for consensus seminorm is shift-invariant along rank-one matrices in~$\eqrowsum^{n \times n}$. That is, for every consensus seminorm~$\gensn{\cdot}:\eqrowsum^{n \times n} \to \R$ and every~$M \in \eqrowsum^{n \times n}$,
\begin{equation}\label{eq:shift}
    \gensn{M + \mathbf{1}c} = \gensn{M}
\end{equation}
for every~$c \in \R^{1 \times n}$.

Several computable submultiplicative consensus seminorms have been proposed previously~\cite{Liu2011DetGossip,DePasquale2024seminorms,Ofir2025ConsensusSeminorms}. In this paper we will mostly be dealing with consensus seminorm in general without restricting interest to any specific form. However, two submultiplicative consensus seminorms will play a special role and are therefore explicitly defined here. The first is the seminorm~$\gensn{M}_\infty : \eqrowsum^{n \times n} \to \R$ defined by
\begin{equation}
    \gensn{M}_\infty = \max_{x \in \mathcal{X}} |Mx|_\infty,
\end{equation}
where~$|x|_\infty \dfb \min_{c \in \R} \|x - \mathbf{1}c\|_\infty$, and~$\mathcal{X} \dfb \{x \in \R^n : |x|_\infty = 1\}$. The second is the seminorm~$\gensn{M}_2 : \eqrowsum^{n \times n}$ defined by
\begin{equation}
    \gensn{M}_2 = \max_{c \in \R^{1 \times n}} \|M - \mathbf{1} c\|_2.
\end{equation}
The special role these seminorms play is due to the following results which appeared in~\cite{Liu2011DetGossip,Ofir2025ConsensusSeminorms,DePasquale2024seminorms}.
\begin{proposition}
    For every stochastic matrix~$S \in \stoc^{n \times n}$,
    \begin{equation}
        \gensn{S}_\infty = \coe(S),
    \end{equation}
    and so~$\gensn{S} < 1$ if and only if~$S$ is scrambling.

    For every doubly stochastic matrix~$S \in \dblstoc^{n \times n}$,
    \begin{equation}
        \gensn{S}_2 = \sigma_2(S),
    \end{equation}
    where~$\sigma_2(S)$ denotes the second largest singular value of~$S$. Furthermore, if all diagonal entries of~$S$ are positive, then~$\gensn{S}_2 < 1$ if and only if the graph of~$S$ is weakly connected.
\end{proposition}

The main important result regarding consensus seminorms is the relation between contraction in a submultiplicative consensus seminorm and exponential convergence to a rank-one matrix.
\begin{theorem}[\cite{Ofir2025ConsensusSeminorms}]\label{thm:contract_gensn}
    Let~$\gensn{\cdot} : \eqrowsum^{n \times n} \to \R$ be a submultiplicative consensus seminorm, and let~$\mathcal{C}$ be a compact subset of matrices in~$\eqrowsum^{n \times n}$ which are all contractive in~$\gensn{\cdot}$ and whose row sums are all equal one. Let
    \begin{equation}\label{eq:gensn_lambda}
        \lambda \dfb \max_{\mathcal{C}} \gensn{M}
    \end{equation}
    Then for each infinite sequence of matrices~$M_1,M_2,\dots$ in~$\mathcal{C}$, the matrix product~$M_i M_{i-1} \cdots M_1$ converges as~$i \to \infty$ to a rank one matrix of the form~$\mathbf{1}c$ as fast as~$\lambda^i$ converges to zero.
\end{theorem}

Let~$P \in \R^{(n-1) \times n}$ be a full-rank matrix whose kernel is the span of~$\mathbf{1}$, denoted~$\consubsp$. For every~$M \in \eqrowsum^{n \times n}$ there exists a unique~$\bar M \in \R^{(n-1) \times (n-1)}$ such that
\begin{equation}\label{eq:Mproj}
    PM = \bar M P.
\end{equation}

The restriction in~\eqref{eq:Mproj} and the shift-invariance implies a relation between consensus seminorms and norms over~$\R^{(n-1) \times (n-1)}$.
\begin{lemma}\label{lem:restnorm}
    Let~$P$ be a full-rank matrix with~$\ker Q = \consubsp$, and let~$\gensn{\cdot} : \eqrowsum^{n \times n} \to \R$ be a consensus seminorm. There exists a norm~$\restnorm{\cdot} : \R^{(n-1) \times (n-1)} \to \R$ such that
    \begin{equation}
        \gensn{M} = \restnorm{\bar M}
    \end{equation}
    for every~$M \in \eqrowsum^{n \times n}$, where~$\bar M \in \R^{(n-1) \times (n-1)}$ is the unique matrix for which~$QM = \bar M Q$.

    Furthermore, if~$\gensn{\cdot}$ is submultiplicative, then so is~$\restnorm{\cdot}$.
\end{lemma}

\begin{proof}
    We will show that the function~$\restnorm{\cdot} : \R^{n-1 \times n -1} \to \R$ given by
    \begin{equation}\label{eq:restnorm}
        \restnorm{A} = \vecsn{M},
    \end{equation}
    where~$M \in \eqrowsum^{n \times n}$ is such that~$PM = AP$, is well-defined. In particular, it will be shown that for every~$A \in \R^{(n-1) \times (n-1)}$ there exist infinitely many~$M \in \eqrowsum^{n \times n}$ such that~$PM = AP$, but that the right hand side of~\eqref{eq:restnorm} is the same for every such~$M$.
    
    Since~$P$ has full row rank, it has a right inverse~$P^{-1} \in \R^{n \times n-1}$. In addition, since~$\ker P = \consubsp$, for every~$A \in \R^{n-1 \times n-1}$,~$P^{-1}AP\mathbf{1} = 0$ so~$P^{-1}AP \in \eqrowsum^{n \times n}$, and~$P(P^{-1}AP) = AP$. Fix~$A \in \R^{n-1 \times n-1}$, and suppose~$M_1,M_2 \in \eqrowsum^{n\times n}$ are such that
    \begin{equation}
        AP = PM_1 = PM_2.
    \end{equation}
    Then, since~$\ker P = \consubsp$, we have that~$M_1-M_2 = \mathbf{1} c$ for some~$c \in \R^{1 \times n}$. It follows from the shift-invariance of consensus seminorms that~$\vecsn{M_1} = \vecsn{M_2}$, so the right hand side of~\eqref{eq:restnorm} is the same for any choice of~$M$.

    It remains to show that~$\restnorm{\cdot}$ is indeed a norm. Let~$A_1,A_2 \in \R^{n-1 \times n-1}$ and~$M_1,M_2 \in \eqrowsum^{n \times n}$ such that~$PM_i=A_iP, i=1,2$. Let~$\alpha \in \R$. Then,~$\alpha(A_1 + A_2)P = \alpha P(M_1 + M_2)$, so
    \begin{equation}
        \restnorm{\alpha(A_1 + A_2)} = \vecsn{\alpha(M_1 + M_2)} \le |\alpha|(\vecsn{M_1} + \vecsn{M_2}),
    \end{equation}
    and this proves homogeneity and the triangle inequality. Nonnegativity of~$\restnorm{\cdot}$ is immediate from the nonnegativity of seminorms. It remains only to show that~$\restnorm{A} = 0$ if and only if~$A = 0$. Since~$P \cdot 0 = 0 \cdot P$,~$\restnorm{0} = 0$. Suppose now that~$A \in \R^{n-1 \times n-1}$ is such that~$\restnorm{A} = 0$. This implies that~$AP = P\mathbf{1}c$ for some~$c \in \R^{1 \times n}$. But, since~$P\mathbf{1}c = 0$, we have that~$APP^{-1} = P\mathbf{1}cP^{-1} = 0$, so~$A = 0$.
    
    Finally, suppose that~$\gensn{\cdot}$ is submultiplicative. Let~$A_1,A_2 \in \R^{n-1 \times n-1}$ and~$M_1,M_2 \in \eqrowsum^{n \times n}$ such that~$PM_i=A_iP, i=1,2$. Then,
    \begin{equation}
        (A_1 A_2)P = A_1 PM_2 = PM_1 M_2,
    \end{equation}
    so
    \begin{equation}
        \restnorm{A_1 A_2} = \vecsn{M_1 M_2} \le \vecsn{M_1}\vecsn{M_2} = \restnorm{A_1}\restnorm{A_2}.
    \end{equation}
\end{proof}

\section{Contraction in dense convergent sets}

This section presents a proof for a more technical but more general result than Theorem~\ref{thm:no_R_K_contraction}.
\begin{theorem}\label{thm:contraction_iff}
    Let~$\mathcal{C} \subset \stoc^{n \times n}$ be a convergent set which is dense in~$\stoc^{n \times n}$. If there exists a submultiplicative consensus seminorm~$\gensn{\cdot} : \eqrowsum^{n \times n} \to \R$ such that~$\gensn{S} < 1$ for all~$S \in \mathcal{C}$, then~$\mathcal{C}$ is a subset of the set of all scrambling matrices~$\stocscr$.
\end{theorem}

\textit{Proof of Theorem~\ref{thm:no_R_K_contraction}:}
    The convergent sets~$\stocpd$ and~$\stocsary$ are dense in~$\stoc^{n \times n}$. Since~$\stocpd \subset \stocsary$ and~$\stocpd$ contains all positive stochastic matrices\footnote{A stochastic matrix is called \emph{positive} if its entries are all positive.}, it is enough to show that the set of all positive stochastic matrices is dense in~$\stoc^{n \times n}$. Fix an arbitrary stochastic matrix~$S$ and an arbitrary positive stochastic matrix~$\bar S$. Note that for every~$\alpha \in (0,1)$, the stochastic matrix
    \begin{equation}
        S(\alpha) = \alpha \bar S + (1-\alpha)S
    \end{equation}
    is positive. Furthermore,
    \begin{equation}
        \|S - S(\alpha)\|_\infty = \|\alpha S - \alpha \bar S\|_\infty = \alpha \|S - \bar S\|_\infty,
    \end{equation}
    so~$\alpha$ can be chosen such that~$S(\alpha)$ is arbitrarily close to~$S$, and this shows that the set of all positive stochastic matrices is dense in~$\stoc^{n \times n}$.

    Finally, for~$n > 2$,~$\stocpd$ contains non-scrambling matrices, see for example Section 2.5 of~\cite{Cao2008ConsensusGraphical}. Therefore, by Theorem~\ref{thm:contraction_iff} there does not exist a submultiplicative consensus seminorm such that~$\gensn{S} < 1$ for all~$S \in \stocpd$, and since~$\stocpd \subset \stocsary$, the same is true for~$\stocsary$, and this proves Theorem~\ref{thm:no_R_K_contraction}.
    \hspace*{\fill}~\QED

\begin{remark}
    The condition in Theorem~\ref{thm:contraction_iff} is not only necessary but also sufficient. That is, if~$\mathcal{C} \subset \stoc^{n \times n}$ is convergent and is a subset of~$\stocscr$, then there exists a submultiplicative consensus seminorm such that~$\gensn{S} < 1$ for all~$S \in \mathcal{C}$. Indeed, the coefficient of ergodicity is such a seminorm.
\end{remark}

\begin{remark}
    The assumption in Theorem~\ref{thm:no_R_K_contraction} that~$n > 2$ is necessary, since for~$n=2$,~$\stocpd \subset \stocsary = \stocscr$, and the coefficient of ergodicity is a submultiplicative consensus seminorm on~$\stocscr$.
\end{remark}

\begin{remark}
    The assumption in Theorem~\ref{thm:contraction_iff} that~$\mathcal{C}$ is dense in~$\stoc^{n \times n}$ is crucial to the conclusion of the theorem. The next section will present convergent sets which are not subsets of~$\stocscr$ and where all matrices are contractive in the same seminorm, but which are not dense in the set of all stochastic matrices.
\end{remark}

The proof of Theorem~\ref{thm:contraction_iff} relies on an interesting property of non-scrambling matrices due to Hajnal, which appeared as Theorem 1 in~\cite{Hajnal1958}.
\begin{lemma}\label{lem:scramb_characterize}
    Let~$X$ be a non-scrambling stochastic matrix. Then there exists a stochastic matrix~$Y$ such that~$XY$ is not convergent.
\end{lemma}
\begin{remark}
    As stated in~\cite{Hajnal1958}, the converse implication is also true: if~$X$ is a scrambling matrix, then for every stochastic matrix~$Y$ the product~$XY$ is scrambling and hence also convergent. This follows immediately from the fact that~$S \in \stoc^{n \times n}$ is such that~$\coe(S) < 1$ if and only if~$S$ is scrambling, and that~$\coe(S_1S_2) \le \coe(S_1) \coe(S_2)$ for every~$S_1,S_2 \in \stoc^{n \times n}$.
\end{remark}
A proof for Lemma~\ref{lem:scramb_characterize} is provided in the appendix.

\textit{Proof of Theorem~\ref{thm:contraction_iff}:}
    Fix~$\mathcal{C} \subset \stoc^{n \times n}$ to be a convergent set which is dense in~$\stoc^{n \times n}$. Let~$\gensn{\cdot} : \eqrowsum^{n \times n} \to \R$ be such that~$\gensn{S} < 1$ for every~$S \in \mathcal{C}$. Since all seminorms over finite dimensional spaces are continuous~\cite{Goldberg2017SeminormContinuous}, this implies that~$\gensn{S} \le 1$ for every stochastic matrix~$S$.

    Seeking a contradiction, suppose that~$\mathcal{C}$ is not a subset of the set of all scrambling matrices, and let~$X \in \mathcal{C}$ be a non-scrambling matrix. Appealing to Lemma~\ref{lem:scramb_characterize}, let~$Y \in \stoc^{n \times n}$ be such that~$XY$ is not convergent. Since~$\gensn{\cdot}$ is bounded above by 1 on all stochastic matrices, and since~$\gensn{\cdot}$ is submultiplicative,
    \begin{equation}
        \gensn{XY} \le \gensn{X}\gensn{Y} = \gensn{X} < 1.
    \end{equation}
    However, Theorem~\ref{thm:contract_gensn} implies that~$XY$ is convergent, and this is a contradiction.
    \hspace*{\fill}~\QED

An attractive feature of the coefficient of ergodicity in~\eqref{eq:coe} is that determining whether~$\coe(S) < 1$ can be done by considering only which entries of~$S$ are positive and which are equal zero, or equivalently by considering only the graph of~$S$. A corollary of Theorem~\ref{thm:contraction_iff} shows that this property is very much unusual among consensus seminorms. We call a consensus seminorm \emph{graphical} if contractivity of a stochastic matrix in the seminorm depends only on the graph of the matrix.
\begin{corollary}
    Let~$\gensn{\cdot}: \eqrowsum^{n \times n} \to \R$ be a graphical submultiplicative consensus seminorm. Then,~$\gensn{S} \ge 1$ for every non-scrambling stochastic matrix~$S \in \stoc^{n \times n}$.
\end{corollary}
\begin{proof}
    Let~$\gensn{\cdot}$ be a graphical submultiplicative consensus seminorm and let~$\mathcal{C} = \{S \in \stoc^{n \times n} : \gensn{S} < 1\}$. Since~$\frac{1}{n} \mathbf{1}\mathbf{1}'$ is a rank one matrix,~$\gensn{\frac{1}{n}\mathbf{1}\mathbf{1}'} = 0$. The matrix~$\frac{1}{n} \mathbf{1}\mathbf{1}'$ is a positive stochastic matrix, so by the assumption that the seminorm is graphical, it follows that~$\gensn{S} < 1$ for every positive stochastic matrix~$S$. As shown in the proof of Theorem~\ref{thm:no_R_K_contraction}, this implies that~$\mathcal{C}$ is dense in~$\stoc^{n \times n}$, so Theorem~\ref{thm:contraction_iff} implies that~$\mathcal{C} \subset \stocscr$, and this completes the proof.
\end{proof}

\subsection{Convergent sets of doubly stochastic matrices}

Interestingly, when considering convergent sets whose matrices all doubly stochastic rather than just stochastic, the results are remarkably different. Recall that~$\stocpd$ is the set of all stochastic matrices with positive diagonal entries and a rooted graph, and~$\dblstocpd$ is the set of all doubly stochastic matrices with positive diagonal entries and a weakly connected graph. Both sets are convergent, and both are maximal in a sense:~$\stocpd$ is the largest convergent set of stochastic matrices with positive diagonal entries, and~$\dblstocpd$ is the largest convergent set of stochastic matrices with positive diagonal entries. Theorem~\ref{thm:no_R_K_contraction} states that there does not exist a submultiplicative consensus seminorm for which all matrices in~$\stocpd$ are contractions. In contrast, all matrices in~$\dblstocpd$ are contractions in the seminorm~$\gensn{\cdot}_2$~\cite{Liu2011DetGossip}.

\section{Joint spectral semiradius}

This section studies the convergence rate in \emph{compact} convergent sets. The discussion is based on the joint spectral radius~\cite{Rota1960JSR}, defined as follows. Let~$\Sigma \subset \R^{n \times n}$ be a compact set of real square matrices. For every~$t > 0$, define~$\Sigma^t \subset \R^{n \times n}$ to be the set consisting of all products of~$t$ matrices from~$\Sigma$, i.e.
\begin{equation}
    \Sigma^t = \{ A_1 \cdots A_t : A_1,\dots,A_t \in \Sigma\}.
\end{equation}
The~\emph{joint spectral radius} is defined as
\begin{equation}\label{eq:jsr}
    \jsr(\Sigma) = \lim_{t \to \infty} \sup_{A \in \Sigma^t} \|A\|^{1/t},
\end{equation}
where~$\|\cdot\|$ denotes an arbitrary matrix norm. It is well-known that the limit is well-defined and that it remains the same regardless of choice of matrix norm~\cite{Jungers2009JSR}.

The main important property of the joint spectral is its close relation to convergence of all infinite products of matrices from the given set~\cite{Daubechies1992RCP}. Specifically, all infinite products of matrices from~$\Sigma$ converge to 0 if and only if~$\jsr(\Sigma) < 1$~\cite{Berger1992Semigroups}, and a similar statement can be made regarding convergence in general rather than convergence to 0 in particular~\cite{Daubechies1992RCP,Berger1992Semigroups}.

The goal of this section is to study a similar construction where the norm in~\eqref{eq:jsr} is replaced with a consensus seminorm. While the joint spectral radius has been used before to study convergence of products of stochastic matrices (see~\cite{Daubechies1992RCP} for the first work in this context and~\cite{Kozyakin2019ConsensusJSR} for a recent review), the definition we pursue here and its implications are novel to the best of our knowledge.

Let~$\mathcal{A} \subset \stoc^{n \times n}$ be a compact set of stochastic matrices. For every~$t > 0$, define~$\mathcal{A}^t$ to be the set consisting of all products of~$t$ matrices from~$\mathcal{A}$,
\begin{equation}
    \mathcal{A}^t \dfb \{A_1 \cdots A_t : A_1,\dots,A_t \in \mathcal{A}\}.
\end{equation}
Let~$\gensn{\cdot} : \eqrowsum^{n \times n} \to \R$ be a submultiplicative consensus seminorm, and define the~\emph{joint spectral semiradius} by
\begin{equation}\label{eq:semijsr}
    \semijsr(\mathcal{A}) = \lim_{t \to \infty} \sup_{A \in \mathcal{A}^t} \gensn{A}^{1/t}.
\end{equation}
Proving that the limit is well-defined and that its value does not depend on the choice of seminorm is done exactly the same as for the standard joint spectral radius and is therefore omitted.%A proof showing that the joint spectral semiradius is well-defined for every submultiplicative consensus seminorm, and that its value does not depend on the particular choice of seminorm, is included in the appendix.

The next result relates the joint spectral radius and the convergence rate in compact convergent sets.
\begin{theorem}\label{thm:semijsr}
    Let~$\mathcal{C} \subset \stoc^{n \times n}$ be a compact set of stochastic matrices. Then~$\mathcal{C}$ is convergent if and only if~$\semijsr(\mathcal{C}) < 1$. Furthermore, if~$\semijsr(\mathcal{C}) < 1$, then for every~$\delta > \semijsr(\mathcal{C})$, there exists a constant~$c > 0$ such that for every sequence~$S_1,S_2,\dots \in \mathcal{C}$,
    \begin{equation}
        \|S_t \cdots S_1 - \mathbf{1} d\|_\infty \le c\delta^t, \quad t=1,2,\dots,
    \end{equation}
    where~$\mathbf{1} d = \lim_{t \to \infty} S_t \cdots S_1$.
\end{theorem}
The proof of Theorem~\ref{eq:semijsr} relies on an auxiliary result, whose proof is given in the appendix.
\begin{lemma}\label{lem:expconverge}
    Let~$\gensn{\cdot} : \eqrowsum^{n \times n} \to \R$ be a submultiplicative consensus seminorm. Let~$S_1,S_2,\dots \in \mathcal{C}$. If there are constants~$c,\delta > 0$ such that
    \begin{equation}\label{cond:semiexpconverge}
        \gensn{S_t \cdots S_1} \le c\delta^t, \quad t=1,2,\dots,
    \end{equation}
    then the infinite product converges to a rank one matrix no slower than~$\delta^t$ converges to zero, i.e.,~$\lim S_t \cdots S_1 = \mathbf{1} d$ for some~$d \in \R^{1 \times n}$, and there are~$\bar c, \delta > 0$ such that
    \begin{equation}
        \|S_t \cdots S_1 - \mathbf{1}d\|_\infty \le \bar c\bar \delta^t, \quad t=1,2,\dots.
    \end{equation}
\end{lemma}
\textit{Proof of Theorem~\ref{thm:semijsr}:}
    Suppose~$\semijsr(\mathcal{C}) \ge 1$. Fix~$P \in \R^{n-1 \times n}$ with full row rank and the span of~$\mathbf{1}$ as its kernel. Let~$\mathcal{\bar C} \subset \R^{n-1 \times n-1}$ be defined by
    \begin{equation}
        \mathcal{\bar C} \dfb \{ \bar S \in \R^{n-1 \times n-1} : \bar SP = PS, S \in \mathcal{C}\}. 
    \end{equation}
    Let~$\restnorm{\cdot} : \R^{n-1 \times n-1} \to \R$ denote the norm defined in Lemma~\ref{lem:restnorm}. Then for every~$t > 0$,
    \begin{multline*}
        \sup \{\gensn{S_t \cdots S_1}^{1/t} : S_1,\dots,S_t \in \mathcal{C}\}
            = \\ \sup \{\restnorm{\bar S_t \cdots \bar S_1}^{1/t} : \bar S_1,\dots,\bar S_t \in \mathcal{\bar C}\}.
    \end{multline*}
    Taking the limit as~$t$ tends to infinity, the left hand side is the joint spectral semiradius of~$\mathcal{C}$ and right hand side is the joint spectral radius of~$\mathcal{\bar C}$. By a well-known result on the joint spectral radius (see for example item (B) of Theorem 1 of~\cite{Berger1992Semigroups}), there exists a sequence~$\bar S_1, \bar S_2, \dots \in \mathcal{\bar C}$ such that the corresponding infinite product does not converge to 0 (and possibly does not converge at all). Let~$S_1, S_2, \dots \in \mathcal{C}$ be such that~$PS_t = \bar S_tP, t=1,2,\dots$. We claim that the infinite product~$\cdots S_3 S_2 S_1$ does not converge to a rank one matrix. Assume by contradiction that~$\lim_{t \to \infty} S_t \cdots S_1 = \mathbf{1}c$ for some~$c \in \R^{1 \times n}$, then
    \begin{align*}
        0 &= \gensn{\lim_{t \to \infty} S_t \cdots S_1}  \\
          &= \lim_{t \to \infty}\gensn{S_t \cdots S_1} \\
          &= \lim_{t \to \infty}\restnorm{\bar S_t \cdots \bar S_1}.
    \end{align*}
    However, this implies that~$\lim_{t \to \infty} \bar S_t \cdots S_1 = 0$, and that is a contradiction.

    Suppose now that~$\semijsr(\mathcal{C}) < 1$ and fix~$\delta \in (\semijsr(\mathcal{C}),1)$. By definition of the joint spectral semiradius, there exists~$s > 0$ such that for any~$t > s$,
    \begin{equation}
        \sup_{S \in \mathcal{C}^t} \gensn{S}^{1/t} \le \delta,
    \end{equation}
    so
    \begin{equation}
        \sup_{S \in \mathcal{C}^t} \gensn{S} \le \delta^t
    \end{equation}
    Then, for all~$t>0$,
    \begin{equation}\label{eq:expbound}
        \sup_{S \in \mathcal{C}^t} \gensn{S} \le \max\{1, C\} \delta^t
    \end{equation}
    where~$C = \max_{r \in \{1,\dots,s\}}\max_{S \in \mathcal{C}^r} \gensn{S}^{1/r} \delta^{-r}$.

    Finally, applying Lemma~\ref{lem:expconverge} proves that there is~$\bar c > 0$ and~$d \in \R^{1 \times n}$ such that
    \begin{equation}
        \|S_t \cdots S_1 - \mathbf{1}d\|_\infty \le \bar c \delta^t,
    \end{equation}
    and this completes the proof.
    \hspace*{\fill}~\QED

An important fact that was shown as part of the proof of Theorem~\ref{thm:semijsr}, is that for every compact convergent set and for every submultiplicative consensus seminorms, there is a finite number~$k$ such that the product of every~$k$ matrices in the convergent set is a contraction.
\begin{corollary}\label{coro:finite_contract}
    Let~$\mathcal{C} \subset \stoc^{n \times n}$ be a compact convergent set. For every submultiplicative consensus seminorm~$\gensn{\cdot} : \eqrowsum^{n \times n} \to \R$, there is a finite number~$t$ such that for every~$S_1,\dots,S_t \in \mathcal{C}$,
    \begin{equation}
        \gensn{S_t \cdots S_1} < 1.
    \end{equation}
\end{corollary}
\vspace{0.5em}

\section{CONCLUSIONS}

This paper discussed the use of seminorms in bounding the rate of convergence in convergent sets of stochastic matrices. A previously published result was refined to show that in some important convergent sets, not all matrices are contractions in the same seminorm. In particular, a necessary and sufficient condition for all matrices in a given convergent set to be contractive in the same seminorm was established under the assumption that the convergent set is dense in the set of all stochastic matrices, which is the case for example if the set includes all positive stochastic matrices. Some convergent sets of interest in applications are not dense in the set of all stochastic matrices, such as the set of all doubly stochastic matrices with positive diagonal entries and a weakly connected graph. It would therefore be useful to derive necessary and sufficient conditions for a general convergent set to be such that all matrices in the set are contractions in the same seminorm.

It was also shown that for every compact convergent set, and every seminorm, there is a finite number~$k$ such that every product of~$k$ matrices from the set is a contraction. For the coefficient of ergodicity, it is well-known that~$k$ is bounded above by a function of the dimension~$n$, and this property has been exploited in previous works to derive bounds on the convergence rate for products of flocking matrices. It is currently unknown whether this property is true for every seminorm. Or, alternatively, whether there are seminorms such that for every finite number~$k$, there is a compact convergent set and a product of~$k$ matrices from the set which is~\emph{not} a contraction.


%%%%%%%%%%%%%%%%%%%%%%%%%%%%%%%%%%%%%%%%%%%%%%%%%%%%%%%%%%%%%%%%%%%%%%%%%%%%%%%%

%%%%%%%%%%%%%%%%%%%%%%%%%%%%%%%%%%%%%%%%%%%%%%%%%%%%%%%%%%%%%%%%%%%%%%%%%%%%%%%%

%%%%%%%%%%%%%%%%%%%%%%%%%%%%%%%%%%%%%%%%%%%%%%%%%%%%%%%%%%%%%%%%%%%%%%%%%%%%%%%%
\section*{APPENDIX}
\subsection{Proof of Lemma~\ref{lem:scramb_characterize}}

This appendix provides a graph-theoretic proof of Lemma~\ref{lem:scramb_characterize}. A directed graph~$\mathbb{G}$ is a tuple~$(\mathcal{V},\mathcal{A})$ of a set of vertices~$\mathcal{V}$ and a set of arcs~$\mathcal{A}$. If there is an arc pointing from vertex~$i$ to vertex~$j$, denoted by~$(i,j) \in \mathcal{A}$, we say vertex~$i$ is a \emph{neighbor} of vertex~$j$, and vertex~$j$ is a \emph{friend} of vertex~$i$. For a given stochastic matrix~$S \in \stoc^{n \times n}$, the graph of~$S$ is the directed graph with vertex set~$\mathbf{n}$ defined such that~$(i,j) \in \mathcal{A}$ if and only if~$s_{ji} > 0$. A directed graph is then called~\emph{stochastic} if it is the graph of some stochastic matrix. Equivalently, a directed graph is stochastic if and only if every vertex of the graph has at least one neighbor. A \emph{walk} from~$i$ to~$j$ is a sequence of vertices starting with~$i$ and ending at~$j$, such that each vertex is a neighbor of the next vertex in the sequence. A~\emph{path} is a walk where no vertex appears more than once. A vertex of a directed graph is called a~\emph{root} if there is a path from this vertex to all vertices in the graph (including itself). A graph is called~\emph{rooted} if it contains at least one root.

Given two graphs~$\mathbb{G}_1, \mathbb{G}_2$ over the same set of vertices, the composition~$\mathbb{G}_2 \circ \mathbb{G}_1$ is the directed graph over the same set of vertices, where~$(i,j)$ is an arc if and only if there exists a vertex~$k$ such that~$(i,k)$ is an arc in~$\mathbb{G}_1$ and~$(k,j)$ is an arc in~$\mathbb{G}_2$. Note that if~$S_1,S_2$ are two stochastic matrices of the same dimensions and~$\mathbb{G}_1, \mathbb{G}_2$ are their graphs, then the graph of the product~$S_2 S_1$ is exactly the composition~$\mathbb{G}_2 \circ \mathbb{G}_1$.

A directed graph with vertex set~$\mathcal{A}$ is called \emph{neighbor-shared} if for every two vertices~$i,j$, there is a vertex~$k$ (possibly equal to~$i$ or~$j$) such that~$(k,i)$ and~$(k,j)$ are both arcs in~$\mathbb{G}$. That is, the graph of a scrambling stochastic matrix is neighbor-shared.

\begin{lemma}\label{lem:neigh}
    Let~$\mathbb{G}$ be a stochastic directed graph which is not neighbor-shared. There exists a stochastic directed graph~$\mathbb{G}'$, over the same set of vertices, such that~$\mathbb{G} \circ \mathbb{G}'$ is not rooted.
\end{lemma}
\begin{proof}
    Fix~$i,j$ to be two vertices without a shared neighbor in~$\mathbb{G}$. Partition the set of vertices of~$\mathbb{G}$ into three subsets~$\mathcal{N}_i,\mathcal{N}_j,\mathcal{\bar N}$: the first consisting of all neighbors of~$i$, the second containing all neighbors of~$j$, and the third containing all vertices which are neither a neighbor of~$i$ nor~$j$. Note that the first two groups are non-empty since~$\mathbb{G}$ is stochastic and are disjoint by the assumption that~$i$ and~$j$ do not share a neighbor. In addition,~$\mathcal{\bar N}$ might be empty, but the union~~$\mathcal{N}_i \cup \mathcal{N}_j \cup \mathcal{\bar N}$ contains all vertices in the graph.

    Let~$\mathbb{G}'$ be the graph over the same vertices as~$\mathbb{G}$, defined such that all vertices in~$\mathcal{N}_i$ have~$i$ as their sole neighbor, all vertices in~$\mathcal{N}_j$ have~$j$ as their only neighbor, and that every vertex in~$\mathcal{\bar N}$ has only itself as a neighbor. Note that~$\mathbb{G}'$ is stochastic by construction, since the union~$\mathcal{N}_i \cup \mathcal{N}_j \cup \mathcal{\bar N}$ contains all vertices in the graph.
    
    Let~$k$ be a neighbor of~$i$ in~$\mathbb{G}$, and note that such a~$k$ exists since~$\mathbb{G}$ is stochastic. Then,~$i$ is a neighbor of~$k$ in~$\mathbb{G}'$, and this proves that~$i$ is a neighbor of itself in~$\mathbb{G} \circ \mathbb{G}'$. To show that~$i$ is the only neighbor of~$i$ in~$\mathbb{G} \circ \mathbb{G}'$, suppose that~$k$ is a neighbor of~$i$ in the composition. Then, there must exist a vertex~$l$ which is a friend of~$k$ in~$\mathbb{G}'$ and a neighbor of~$i$ in~$\mathbb{G}$. In particular,~$l \in \mathcal{N}_i$, but then by construction the only friend of~$l$ in~$\mathbb{G}'$ is~$i$, so~$k=i$. Repeating the same argument for~$j$ shows that the only neighbor of~$j$ in~$\mathbb{G} \circ \mathbb{G}'$ is~$j$ itself.
    
    We conclude that~$i$ is the only neighbor of~$i$ in~$\mathbb{G} \circ \mathbb{G}'$, and that~$j$ is the only neighbor of~$j$. Hence,~$\mathbb{G} \circ \mathbb{G}'$ cannot be rooted.
\end{proof}

\begin{lemma}\label{lem:convergent_is_rooted}
    Let~$S \in \stoc^{n \times n}$ and let~$\mathbb{G}$ be its graph. If~$S$ is convergent, then~$\mathbb{G}$ is rooted. 
\end{lemma}
\begin{proof}
    Suppose~$S$ is convergent and let~$c \in \R^{1 \times n}$ be such that~$\lim_{t \to \infty} S^t = \mathbf{1}c$. Since~$S$ is stochastic,~$c$ has at least one positive entry, say at index~$i$. Then, then there is some~$t_0$ such that for all~$t > t_0$, all entries of column~$i$ of~$S^t$ are positive.
    
    In other words, if~$\mathbb{G}^t$ is the~$t$-fold composition of~$\mathbb{G}$ with itself, then for all~$t > t_0$, vertex~$i$ is a neighbor of all vertices in~$\mathbb{G}^t$. Since~$(i,j)$ is an arc in~$\mathbb{G}^t$ if and only if there is a walk of length~$t$ from~$i$ to~$j$ in~$\mathbb{G}$, we conclude that there is a path from~$i$ to every vertex in~$\mathbb{G}$, so~$\mathbb{G}$ is rooted.
\end{proof}

\vspace{1em}

\textit{Proof of Lemma~\ref{lem:scramb_characterize}:}
    Let~$X \in \stoc^{n \times n}$ be a stochastic matrix which is not scrambling, and let~$\mathbb{G}$ be its graph. Let~$\mathbb{G}'$ be the stochastic graph given by Lemma~\ref{lem:neigh}, and~$Y \in \stoc^{n \times n}$ be a stochastic matrix whose graph is~$\mathbb{G}'$. Then, by Lemma~\ref{lem:neigh}, the graph of the product~$XY$ is not rooted, so by Lemma~\ref{lem:convergent_is_rooted}~$XY$ is not convergent.
    \hspace*{\fill}~\QED

\subsection{Proof of Lemma~\ref{lem:expconverge}}

    Let~$P \in \R^{n-1 \times n}$ be a full rank matrix whose kernel is the span of~$\mathbf{1}$. Fix an arbitrary sequence~$S_1,S_2,\dots,$ in~$\mathcal{C}$, and let~$\bar S_1, \bar S_2, \dots \in \R^{n-1 \times n-1}$ be such that
    \begin{equation}
        PS_i = \bar S_i P, \quad i=1,2,\dots
    \end{equation}
    Let~$Q \in \R^{1 \times n}$ be such~$T = \begin{bmatrix} P' Q' \end{bmatrix}$ is invertible, then
    \begin{equation}
        TS_iT^{-1} = \begin{bmatrix}
            1 & \hat S_i \\
            0 & \bar S_i
        \end{bmatrix}
    \end{equation}
    for some~$\hat S_1, \hat S_2, \dots \in \R^{1 \times n-1}$. It may be verified that
    \begin{equation}
        TS_t \cdots S_1T^{-1} = \begin{bmatrix}
            1 & \hat S_1 + \sum_{i=2}^t \hat S_i \bar S_{i-1} \cdots \bar S_1 \\
            0 & \bar S_t \cdots \bar S_1
        \end{bmatrix}.
    \end{equation}
    Let~$\restnorm{\cdot} : \R^{n-1 \times n-1} \to \R$ be the norm defined in Lemma~\ref{lem:restnorm}. Then~\eqref{cond:semiexpconverge} implies that
    \begin{equation}
        \restnorm{\bar S_t \cdots \bar S_1} \le c \delta^t, \quad t=1,2,\dots,
    \end{equation}
    which by equivalence of norms implies that there is~$\bar c > 0$ such that
    \begin{equation}\label{eq:bar_infty_exp}
        \|\bar S_t \cdots \bar S_1\|_\infty \le \bar c \delta^t, \quad t=1,2,\dots
    \end{equation}
    We claim that this implies the convergence of~$u(t) = \hat S_1 + \sum_{i = 2}^t \hat S_i \bar S_{i-1} \cdots \bar S_1$. Indeed, fix~$r>s>0$, then
    \begin{align*}
        \|u(r) - u(s)\|_\infty &= \|\sum_{i=s+1}^r \hat S_i \bar S_{i-1} \cdots \bar S_1 \|_\infty \\
            &\le \sum_{i=s+1}^r \| \hat S_i \|_\infty \|\bar S_{i-1} \cdots \bar S_1 \|_\infty,
    \end{align*}
    where the inequality is due to the triangle inequality and submultiplicativity. Since~$\mathcal{C}$ is compact, the sequence~$\hat S_1, \hat S_2, \dots$ is bounded. Let~$\hat c = \sup_i \|\hat S_i\|_\infty$, then
    \begin{align*}
        \|u(r) - u(s)\|_\infty &\le  \hat c \sum_{i=s+1}^r \|\bar S_{i-1} \cdots \bar S_1 \|_\infty \\
        &\le \hat c \bar c \sum_{i=s+1}^r \delta^i.
    \end{align*}
    Finally, solving the sum on the right hand side gives
    \begin{equation}
        \|u(r) - u(s)\|_\infty \le \hat c \bar c \frac{\delta^{s+2} - \delta^{r+1}}{1-\delta} = \hat c \bar c \delta^s\frac{\delta^2 - \delta^{r-s+1}}{1-\delta}.
    \end{equation}
    Since~$r > s$ and~$0 < \delta < 1$,~$\delta^2 - \delta^{r-s+1} \le 1$, so
    \begin{equation}
        \|u(r) - u(s)\|_\infty \le \hat c \bar c\delta^s\frac{1}{1-\delta},
    \end{equation}
    and this shows that the sequence~$u(1),u(2),\dots$ is a Cauchy sequence and therefore convergent. Let~$u^* = \lim_{t \to \infty} u(t)$, then for every~$t > 0$,
    \begin{align}\label{eq:u_exp}
        \|u^* - u(t)\|_\infty &= \lim_{s \to \infty}\|u(s) - u(t)\|_\infty 
            \le \hat c \bar c \delta^{t} \frac{1}{1-\delta}.
    \end{align}
    Let~$d \in \R^{1 \times n}$ be such that
    \begin{equation}
        \mathbf{1}d = T^{-1}\begin{bmatrix}
            1 & u^* \\
            0 & 0
        \end{bmatrix}T.
    \end{equation}
    Combining~\eqref{eq:u_exp} and~\eqref{eq:bar_infty_exp} yields
    \begin{align*}
        \|S_t \cdots S_1 - \mathbf{1}d\|_\infty &\le \|T^{-1}T(S_t \cdots S_1 - \mathbf{1}d)T^{-1}T\|_\infty \\
        &\le \|T\|_\infty\|T^{-1}\|_\infty \left\| \begin{bmatrix}
            0 & u(t) - u^* \\ 0 & \bar S_t \cdots \bar S_1
        \end{bmatrix} \right\|_\infty \\
        &= \|T\|_\infty\|T^{-1}\|_\infty \max\{ \hat c \bar c \delta^t \frac{1}{1-\delta}, \bar c \delta^t \} \\
        &= \|T\|_\infty\|T^{-1}\|_\infty \max\{ \hat c \bar c \frac{1}{1-\delta}, \bar c \}  \delta^t,
    \end{align*}
    and this completes the proof.
    \hspace*{\fill}~\QED

%\section*{ACKNOWLEDGMENT}

%%%%%%%%%%%%%%%%%%%%%%%%%%%%%%%%%%%%%%%%%%%%%%%%%%%%%%%%%%%%%%%%%%%%%%%%%%%%%%%%

\bibliographystyle{IEEEtranS}
\bibliography{literature}

\end{document}